\documentclass[journal,twoside,web]{ieeecolor}

\usepackage{generic}
\usepackage{amsmath,amssymb,amsfonts}
\usepackage{algorithmic}
\usepackage[caption=false,font=footnotesize,labelformat=simple]{subfig}
\usepackage{graphicx}
\usepackage{textcomp}

\usepackage{epsfig} 
\usepackage{lcsys}
\usepackage{booktabs}
\usepackage{multirow}

\usepackage{bm}  
\usepackage[hidelinks]{hyperref}
\usepackage{mathtools} 

\usepackage[style=ieee,citestyle=numeric-comp,]{biblatex} 
\AtEveryBibitem{\clearfield{url}}
\AtEveryBibitem{\clearfield{doi}}
\AtEveryBibitem{\clearfield{isbn}}
\AtEveryBibitem{\clearfield{issn}}
\AtBeginBibliography{\small}

\DeclareMathOperator*{\argmin}{arg\,min}

\usepackage{amsthm}

\newtheoremstyle{boldstyle}%
{}
{}
{\itshape}
{}
{\bfseries}
{.}
{.5em}
{\thmname{#1}\thmnumber{ #2}\thmnote{ (#3)}}
\theoremstyle{boldstyle}
\newtheorem{theorem}{Theorem}[section]

\newtheorem{lemma}[theorem]{Lemma}
\newtheorem{proposition}[theorem]{Proposition}

\newtheorem{remark}{Remark}

\usepackage{eso-pic}

\newcommand{\publishernotetext}{%
	\textcopyright{} 2026 IEEE. Personal use of this material is permitted.  Permission from IEEE must be obtained for all other uses, in any current or future media, including reprinting/republishing this material for advertising or promotional purposes, creating new collective works, for resale or redistribution to servers or lists, or reuse of any copyrighted component of this work in other works.%
}

\newcommand{\copyrightnotice}{%
	\AddToShipoutPictureFG*{%
		\AtPageLowerLeft{%
			\raisebox{7mm}{%
				\makebox[\paperwidth][c]{%
					\fcolorbox{black}{white}{%
						\parbox{0.842\paperwidth}{%
							\centering\footnotesize
							\publishernotetext
						}%
					}%
				}%
			}%
		}%
	}%
}

\graphicspath{ {./Pics/} }

\begin{document}

\copyrightnotice

\def\BibTeX{{\rm B\kern-.05em{\sc i\kern-.025em b}\kern-.08em
    T\kern-.1667em\lower.7ex\hbox{E}\kern-.125emX}}

\title{Where to Put Safety? Control Barrier Function Placement in Networked Control Systems}

\author{Severin Beger, Yuling Chen, and Sandra Hirche
\thanks{The authors would like to thank the Federal Ministry of Research, Technology, and Space (BMFTR) for its support as part of the research program Communication Systems “Souverän. Digital. Vernetzt.”. Joint project 6G-life, project identification number: 16KIS2414}
\thanks{All authors are with the Chair of Information-oriented Control, TU Munich, 80333 Munich, Germany (E-mail: \{Severin.Beger\},\{Yuling1.Chen\}, \{Hirche\}@tum.de). }
}

\maketitle
\thispagestyle{empty}

\begin{abstract}
Control barrier functions (CBFs) are widely used to enforce safety in autonomous systems, yet their placement within networked control architectures remains largely unexplored.
In this work, we investigate where to enforce safety in a networked control system in which a remote model predictive controller (MPC) communicates with the plant over a delayed network. We compare two safety strategies: i) a local myopic CBF filter applied at the plant and ii) predictive CBF constraints embedded in the remote MPC. For both architectures, we derive state-dependent error tolerance bounds and show that safety placement induces a fundamental trade-off: local CBFs provide higher disturbance tolerance due to access to fresh state measurements, whereas MPC-CBFs achieve better performance through anticipatory behavior, but yield stricter admissible disturbance levels. Motivated by this insight, we propose a combined architecture that integrates predictive and local safety mechanisms. The theoretical findings are illustrated in simulations on a planar three-degree-of-freedom robot performing a collision-avoidance task.
\end{abstract}
\begin{IEEEkeywords}
Control Barrier Functions, Control System Architecture, Networked Control Systems
\end{IEEEkeywords}

\section{Introduction}
\label{sec:introduction}
\IEEEPARstart{S}{afety} is paramount for cyber-physical systems such as autonomous vehicles operating in public spaces, or robots interacting with humans. Control barrier functions (CBFs) have  emerged as a powerful framework for enforcing safety constraints by guaranteeing forward invariance of safe sets while minimally modifying nominal control inputs \cite{Wieland.2007, A.D.Ames.2017, A.D.Ames.2019}.\\
CBFs are commonly implemented as myopic, single-step safety filters acting at the plant. While effective in ensuring safety, purely local safety filters may lead to undesirable behavior such as late obstacle avoidance, loss of feasibility, or deadlock situations \cite{K.P.Wabersich.2023,A.Singletary.2021}. To address these issues, recent work has integrated CBF constraints directly into model predictive control (MPC) formulations \cite{J.Zeng.2021}, allowing safety to be considered within a predictive optimization framework.\\
In many modern cyber-physical systems, including cloud robotics, autonomous vehicle fleets, and teleoperated robots, control is implemented in a networked architecture \cite{Matni.2024}, where computationally intensive planning or optimization is performed remotely and low-level feedback control is executed locally. While remote optimization enables complex predictive planning, it relies on delayed or predicted state information due to communication constraints \cite{Beger.2024}.
This raises a fundamental architectural question: where should safety mechanisms such as CBFs be placed in a networked control system (NCS)? \\
CBFs can either be enforced locally at the plant side, benefiting from fresh measurements but acting myopically, or embedded in a remote predictive controller, enabling anticipatory behavior but relying on predicted states affected by network delay and disturbances. While recent work has considered potential architectures for CBF in networked \cite{R.Periotto.2024} or multi-rate \cite{U.Rosolia.2021} setups, the architectural implications of enforcing safety locally or remotely have not yet been systematically analyzed.\\
In this paper, we investigate the placement of CBF-based safety filters in a networked control architecture consisting of a remote MPC and a locally executed controller. More specifically, we compare two safety strategies: i) a local myopic CBF filter applied at the plant and ii) predictive CBF constraints embedded in the remote MPC. We leverage the fact that a discrete-time CBF constraint is evaluated using a model-predicted successor state. Thus, its safety certificate is only as reliable as the prediction on which it is based. To this end, we derive a prediction-error tolerance condition for discrete-time CBFs, i.e., the amount of mismatch between predicted and true states that can be absorbed while still certifying safety, and specialize it to local and remote CBF placement subject to disturbances. Through formal comparison, we reveal an architectural trade-off between disturbance tolerance and predictive performance. Based on these findings, we propose a hybrid safety architecture combining predictive and local safety mechanisms.
Our results are illustrated in simulations on a planar three-degrees-of-freedom robot performing a collision-avoidance task.\\
The paper is organized as follows. Sec. \ref{sec:problem} introduces the problem setup. Sec. \ref{sec:ErrorMargin} introduces a prediction error margin for discrete CBFs, which we use in Sec. \ref{sec:theory} to compare the specific placement of CBFs in an NCS. In Sec. \ref{sec:CombinedArchitecture}, we propose an approach to combine the predictive and myopic safety components before we present simulation results in Sec. \ref{sec:sim}. Finally, we conclude the work in Sec. \ref{sec:conclusion}.

\section{Problem Statement and Preliminaries}
\label{sec:problem}
In this work, we consider discrete-time, nonlinear dynamical systems subject to additive disturbance 
\begin{align}
\label{eq:dynamics}
   \bm{x}_{k+1} = \bm{f}(\bm{x}_k,\bm{u}_k)+\bm{w}_k, \quad k\in\mathbb{N}
\end{align}
with states $\bm{x}_k \in \mathbb{R}^n$, inputs $\bm{u}_k \in \mathbb{R}^m$, and disturbances $\bm{w}_k\in \mathcal{W}\subset \mathbb{R}^n$ that satisfy $\lVert\bm{w}_k\rVert \leq \bar{w}$. $\bm{f}(\bm{x},\bm{u}): \mathbb{R}^{n\times m} \to \mathbb{R}^n$ is locally Lipschitz with constants $L_x$ and $L_u$. 
All states and inputs are subject to admissible sets $\bm{x}_k \in \mathcal{X}$ and $\bm{u}_k \in\mathcal{U}$ such that $\mathcal{U}\coloneqq\{\bm{u} \in \mathbb{R}^m:\lVert\bm{u}\lVert \leq\bar{u}\}$ for all $k$.\\
A communication network between the controller and the plant introduces a measurement-to-actuation latency $\tau_k = \hat{\tau} + r_{k} \in \mathbb{N}$. While physical network delays can be decomposed into sensor-to-controller delay $\tau_{sc}$, computational delay $\tau_c$ and controller-to-actuator delay $\tau_{ca}$, we consider a compensation-oriented decomposition (cf. Fig. \ref{fig:TimingOverview}). We differentiate a known, constant delay part $\hat{\tau} \in \mathbb{N}$ and an unknown, uncompensated delay residual $r_{k} \in \mathbb{N}_0$ upper-bounded by $0 \leq r_k \leq \bar{r}$. Accordingly, we distinguish the following three points of time in the network: the time of a measurement $k-\hat{\tau}$, the predicted time of control application $k$, and the true control application time $k+r_k$. All time stamps are taken at the plant. \\
\begin{figure}[t]
\centering
\includegraphics[width=\linewidth]{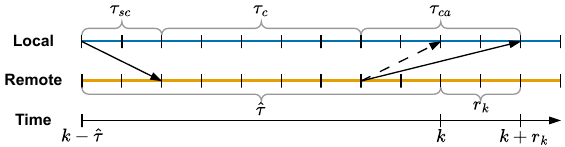}
\caption {Timing of the considered networked communication. Solid arrows symbolize the actual timing of sent messages, while the dashed arrow represents a predicted timing.}
\label{fig:TimingOverview}
\end{figure}
 On the remote side of the system (cf. Fig. \ref{fig:SystemOverview}), a networked predictive controller (NPC) receives a delayed state measurement $\bm{x}_{k-\hat{\tau}}$. To compensate for the delay, a predictor rolls out the nominal dynamics
\begin{align}
\label{eq:nominalDynamics}
    \hat{\bm{x}}_{k+1} =\bm{f}(\hat{\bm{x}}_k,\bm{u}_k)
\end{align}
initialized with $\bm{x}_{k-\hat{\tau}}$ and the last $\hat{\tau}$ buffered inputs from the remote buffer $\mathcal{B}_\mathrm{r}$.
The resulting state estimate $\hat{\bm{x}}_{k}$ serves as the initial value $\bm{x}_{0|k}=\hat{\bm{x}}_{k}$ for the optimal control problem (OCP) of an MPC
\begin{subequations}
\label{eq:OCP}
    \begin{minipage}{0.98\linewidth}
    \vspace{0.8ex}
    \hrule
    \vspace{0.8ex}
    \textbf{Nominal MPC:}
    \vspace{0ex}
        \begin{align}
                J= &\min_{\bm{u}_{0:N-1|k}} p(\bm{x_{N|k}})+\sum_{i=0}^{N-1}q(\bm{x}_{i|k},\bm{u}_{i|k})\\
              \textrm{s.t. }      
             &\bm{x}_{i+1|k} = \bm{f}(\bm{x}_{i|k},\bm{u}_{i|k}),\,\label{eq:OCP-dynConst}\\
              &\bm{x}_{i|k} \in \mathcal{X}, \quad \bm{x}_{0|k} = \hat{\bm{x}}_{k}, \quad\bm{x}_{N|k} \in \mathcal{X}_f\label{eq:OCP-stateConst},\\
              &\bm{u}_{i|k} \in \mathcal{U},\quad \lVert \bm{u}_{i+1|k}-\bm{u}_{i|k}\rVert \leq \Delta u,\label{eq:OCP-inputConst}\\
            &\forall i=\{0,\ldots,N-1\} \nonumber
            ,
        \end{align}
    \vspace{-2.5ex}
    \hrule
    \vspace{1.5ex}
    \end{minipage}
\end{subequations}
where $\Delta u$ is a rate constraint, and $\mathcal{X}_f$ the terminal set. $p(\bm{x}_N)\in\mathbb{R}^+$ denotes the terminal costs at prediction horizon $N\in\mathbb{N}_{>0}$, while $q(\bm{x}_{i|k},\bm{u}_{i|k}) \in \mathbb{R}^+$ represents the stage costs. $\bm{x}_{i|k}$ denotes the predicted state vector at time step $k+i$ based on the initial value $\hat{\bm{x}}_k$ and applying the optimal input sequence $\bm{u}_{0:N-1|k}$ to the nominal dynamics \eqref{eq:nominalDynamics}. Note that $\hat{\bm{x}}_{k+l}=\bm{x}_{l|k}$. Feasibility of the problem is enforced through standard choices for the terminal ingredients \cite{Mayne.2014}.\\ 
The first entry of the resulting input array 
\begin{align}
\label{eq:MPCControllaw}
    \bm{u}^\mathrm{MPC}_{k}=\bm{u}_{0|k}(\hat{\bm{x}}_{k}).
\end{align}
is buffered in $\mathcal{B}_\mathrm{r}$, and sent to the plant side.\\
\begin{figure}[t]
\centering
\includegraphics[width=.7\linewidth]{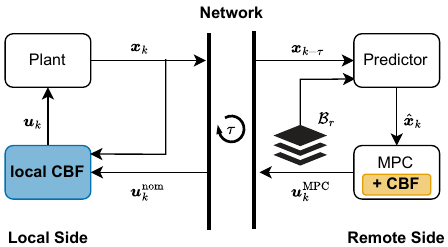}
\caption{Considered networked setup with potential CBF placements locally or remotely within the MPC.}
\label{fig:SystemOverview}
\end{figure}
The main concern of this work is to render the system safe with respect to a safe set 
\begin{align}
\label{eq:SafeSet}
    \mathcal{C}=\{\bm{x}_k\in\mathcal{X}\subset\mathbb{R}^n:h(\bm{x}_k)\geq0\}.
\end{align}
where $h(\bm{x}_k): \mathcal{X} \to \mathbb{R}$ is the Lipschitz continuous barrier function with constant $L_h$. We consider the discrete-time CBF condition, as introduced in \cite{Agrawal.2017}, 
\begin{align}
\label{eq:CBFCondition}
    h(\bm{x}_{k+1})-h(\bm{x}_k) \geq -\gamma(h(\bm{x}_k)), \,\,\gamma \in \mathcal{K}^e_\infty
\end{align}
to enforce forward invariance of the safe set.
For the sake of simplicity, we treat $\gamma$ in the rest of this article as a scalar $0<\gamma\leq1$. This results in the well-known exponential CBF, as \eqref{eq:CBFCondition} can be rearranged to give $h(\bm{x}_{k+1})\geq (1-\gamma)h(\bm{x}_k)$, which describes exponential decay of the barrier function.\\
But where do we enforce the set invariance? Either, we incorporate the CBF certificate \eqref{eq:CBFCondition} as an additional constraint into the MPC problem \eqref{eq:OCP}:
\begin{equation}
    \begin{minipage}{0.98\linewidth}
    \vspace{0.8ex}
    \hrule
    \vspace{0.8ex}
    \textbf{Remote MPC-CBF:}
    \vspace{0ex}
            \begin{align}
            \label{eq:MPC-CBF} 
              &\textrm{Problem \eqref{eq:OCP} with additional constraints:} \\
            &h(\bm{x}_{i+1|k})-h(\bm{x}_{i|k}) \geq -\gamma h(\bm{x}_{i|k}),
            \label{eq:mpc-cbf-barrier}\nonumber\\
            & \forall i=\{0,\ldots,N-1\}. \nonumber
            \end{align}
    \vspace{-2.5ex}
    \hrule \nonumber
    \vspace{1.5ex}
    \end{minipage}
\end{equation}
Thus, we consider the safety in a predictive manner at the downside of using outdated state measurements. Or alternatively, we place a myopic filter at the local side close to the plant, which acts on the most recent state measurements and alters the nominal input from the MPC $\bm{u}^\mathrm{nom}_k= \bm{u}^\mathrm{MPC}_k$ before applying it to the plant to enforce safety:
\begin{subequations}
\label{eq:localCBF} \nonumber
    \begin{minipage}{0.98\linewidth}
    \vspace{0.8ex}
    \hrule
    \vspace{0.8ex}
    \textbf{Local CBF:}
    \vspace{1.2ex}
            \begin{align}
            \bm{u}_k = &\argmin_{\bm{u}} \,\frac{1}{2}\lVert \bm{u}-\bm{u}^\mathrm{nom}_k \rVert^2
            \label{eq:localCBF-costs}\\
           &\textrm{s.t. } h(\hat{\bm{x}}_{k+1})-h(\bm{x}_k) \geq -\gamma h(\bm{x}_k),
            \label{eq:localCBF-barrierconstraint}\\
            &\bm{u}_k \in \mathcal{U}.
            \label{eq:localCBF-inputconstraints}
            \end{align}
    \vspace{-2.5ex}
    \hrule
    \vspace{1.5ex}
    \end{minipage}
\end{subequations}
The different placements are depicted in Fig. \ref{fig:SystemOverview}.\\
Note that either way, the resulting optimization problem is in general nonlinear, even for control-affine dynamics, due to the inclusion of the dynamic flow in the control barrier function in the discrete setting \cite{Agrawal.2017}. Note that the tuning of $\gamma$ plays a crucial role \cite{J.Zeng.2021}, as it affects the conservativeness of the discrete CBF and its compatibility with the feasible set of the MPC.\\
To analyze how these two architectures trade off in terms of performance and safety under disturbance, we introduce the viewpoint of prediction error margin in the following section.

\section{Error Margin of Discrete CBFs}
\label{sec:ErrorMargin}
When applying the CBF condition \eqref{eq:CBFCondition} to choose input $\bm{u}_k$, the state $\bm{x}_{k+1}$ is not available. Hence, the CBF condition is evaluated using the predicted state $\hat{\bm{x}}_{k+1}$. In a networked predictive architecture, the state from which the CBF condition is initialized may also be predicted. Consequently, the CBF certificate is computed from predicted states, while the safety has to hold for the true states.
Let $\bm{e}_k = \bm{x}_k - \hat{\bm{x}}_k$ denote the prediction error for step $k$. Then, using Lipschitz continuity of $h$, we can lower bound $h(\bm{x}_{k+1})$ as
\mbox{$h(\bm{x}_{k+1})\geq h(\hat{\bm{x}}_{k+1})-L_h \lVert \bm{e}_{k+1} \lVert$}
while we can upper bound $h(\bm{x}_k)$ as
\mbox{$h(\bm{x}_{k}) \leq h(\hat{\bm{x}}_{k})+L_h \lVert \bm{e}_{k} \lVert$}.
Applying both bounds to the CBF condition \eqref{eq:CBFCondition} gives
\begin{align}
    \label{eq:PredictionErrorCondition}
    \frac{h(\hat{\bm{x}}_{k+1})-(1-\gamma)h(\hat{\bm{x}}_k)}{L_h} \geq \lVert\bm{e}_{k+1}\lVert+(1-\gamma)\lVert\bm{e}_k\lVert.
\end{align}
As long as the loss caused by prediction errors is smaller than the prediction-error margin  
\begin{align}
\label{eq:PredErrorTolerance}
    \Delta_k(\hat{\bm{x}}_k,\bm{u}_k) = \frac{h(\hat{\bm{x}}_{k+1})-(1-\gamma)h(\hat{\bm{x}}_k)}{L_h}.
\end{align}
the CBF condition holds for the true state.
Relating $\Delta_k$ to the error generated by a specific architecture yields the corresponding \emph{disturbance tolerance}. Thus, tolerance increases with the available CBF margin and decreases with the barrier function's sensitivity $L_h$.
\begin{remark}[Sources of Prediction Errors]
The errors $\bm{e}_{k}$ and $\bm{e}_{k+1}$ collect all deviations between the states used in the CBF constraint and the true states. In this work, these deviations result from additive disturbances, delay-compensating prediction, and uncompensated delay. Other effects, not considered here for simplicity, can be included in the same way.
\end{remark}
This viewpoint of prediction error margin lets us separate the CBF certificate from the architecture that generates the prediction errors. The local and remote safety placements differ precisely in these errors. We now specialize the tolerance condition to the two architectures considered in this paper.

\section{Disturbance Tolerance of Locally and Remotely Placed CBFs}
\label{sec:theory}
In the following sections, we explore the inherent error tolerance of the considered architectures with respect to the additive disturbance bound while preserving safety, and examine how they trade off. We start with the local, myopic version.
\subsection{Local CBF}
We define the available CBF margin for a myopic filter with access to the current state $\bm{x}_k$ as
\begin{align}\label{eq:localCBFPredictionError}
        \Delta_k^{\mathrm{loc}}(\bm{x}_k, \bm{u}_k)\coloneqq \frac{h(\hat{\bm{x}}_{k+1})-(1-\gamma)h(\bm{x}_k)}{L_h}.
\end{align}
\begin{lemma}[Disturbance Tolerance of Local CBF]
\label{lem:MyopiCBF}
Consider system \eqref{eq:dynamics} with bounded disturbance $\lVert\bm{w}_k\rVert \leq \bar{w}$, safe set
$\mathcal C$ \eqref{eq:SafeSet}, and a local CBF \eqref{eq:localCBF} with access to the current state $\bm{x}_k$ and the nominal input $\bm{u}^\mathrm{nom}_k$. Let $h(\cdot)$ be Lipschitz continuous with constant $L_h$ and $h(\bm{x}_0)\geq0$.
If, for every closed-loop state $x_k\in\mathcal C$, the local CBF is feasible and returns $u_k\in\mathcal U$ such that
\begin{align}
\label{eq:myopicNoiseCondition}
      \bar{w} \leq \bar{w}_l(\bm{x}_k,\bm{u}_k)\coloneqq\Delta_k^{\mathrm{loc}}(\bm{x}_k, \bm{u}_k)
\end{align}
then it renders the safe set $\mathcal{C}$ forward invariant.
\end{lemma}
\begin{proof}
    $\Delta_k^\mathrm{loc}(\bm{x}_k, \bm{u}_k)$ follows from \eqref{eq:PredErrorTolerance} and $\hat{\bm{x}}_k=\bm{x}_k$ in the local case. The latter condition also results in $\bm{e}_k=\bm{0}$. Therefore, the one-step prediction error
    $\bm{e}_{k+1} = \bm{w}_k$
    provides the upper bound $\lVert \bm{w}_k \lVert \leq \bar{w}$ for the tolerable prediction error. Thus, if \eqref{eq:myopicNoiseCondition} is true, then \eqref{eq:PredictionErrorCondition} with $\bm{e}_k=0$ and $\hat{\bm{x}}_k = \bm{x}_k$ holds, and therefore
$h(\bm{x}_{k+1})\geq(1-\gamma)h(\bm{x}_k)$. Since $\bm{x}_0\in\mathcal C$ and $0< \gamma \leq 1$, recursive application gives $h(x_k)\geq 0 \; \forall k$.
\end{proof}
The bound decreases near the boundary $\partial\mathcal C$, where the nominal successor must move sufficiently inside the safe set to compensate the disturbance-induced barrier loss. Robust safety can be enforced by tightening the CBF constraint with $L_h \bar{w}$ (cf. \cite{Liu.2026}).

\subsection{Remote MPC-CBF}
Next, we examine the tolerated state-dependent disturbance bound of the remote MPC with CBF constraints \eqref{eq:MPC-CBF}, where $\bm{u}_k = \bm{u}_k^{\mathrm{MPC}}$, since no local filter alters the inputs.
Let \mbox{$S_r:=\sum_{j=0}^{r-1}L_x^j$}, \mbox{$d_u:=L_u\bar{r}\Delta_u$}, and \mbox{$G(\hat{\tau},l)=S_{\hat{\tau}+l}+(1-\gamma)S_{\hat{\tau}+l-1}$}. We define the available CBF margin of \eqref{eq:MPC-CBF} recursively along a horizon as 
\begin{align}\label{eq:MPCCBFPredictionError}
        \Delta^\mathrm{rem}_{k+l}(\hat{\bm{x}}_k,\bm{u}_k,\hat{\tau},l)\coloneqq \frac{h(\hat{\bm{x}}_{l+k})-(1-\gamma) h(\hat{\bm{x}}_{l-1+k})}{L_h}.
\end{align}
\begin{lemma}[Disturbance Tolerance of Remote MPC-CBF]
\label{lem:MPCCBF}
Consider system \eqref{eq:dynamics} with bounded disturbance $\lVert\bm{w}_k\rVert \leq \bar{w}$, safe set
$\mathcal C$ \eqref{eq:SafeSet}, and the remote MPC-CBF \eqref{eq:MPC-CBF} connected through a network with measurement-to-actuation latency \mbox{$\tau_k = \hat{\tau} + r_k \in \mathbb{N}$}. Let $h(\cdot)$ be Lipschitz continuous with constant $L_h$ and $h(\bm{x}_0)>0$. 
If, for every closed-loop state $x_k\in\mathcal C$,  the MPC-CBF is feasible and computes a $\bm{u}_k \in \mathcal{U}$ such that
\begin{align}
    &\bar{w} \leq \bar{w}_r(\hat{\bm{x}}_k,\bm{u}_k,\hat{\tau},\bar{r},l) \coloneqq \frac{\Delta^\mathrm{rem}_{k+l}(\hat{\bm{x}}_k,\bm{u}_k,\hat{\tau},l)}{G(\hat{\tau},l)} - d_u \label{eq:MPCCBFNoiseCondition}
\end{align}
for all \mbox{$l \in \{1,\dots,N\}$}, then it renders the safe set $\mathcal{C}$ forward invariant.
\end{lemma}
\begin{proof}
Due to delayed measurements, the MPC is initialized with a predicted state $\hat{\bm{x}}_k$, resulting in
\begin{align}
    \bm{e}_{k+1} = \bm{f}(\bm{x}_{k},\bm{u}_{k-r_k}) + \bm{w}_k-\bm{f}(\hat{\bm{x}}_{k},\bm{u}_{k}).
\end{align}
The error due to the uncompensated delay between the applied and predicted inputs satisfies $\lVert \bm{u}_{k-r_k}-\bm{u}_k\lVert\leq \bar{r}\Delta_u$, where we have considered the rate constraint of the MPC. Hence, bounding the prediction error gives
\begin{align}
\label{eq:PredictionErrorBoundRemote}
    \lVert \bm{e}_{k+1}\lVert \leq L_x \lVert\bm{e}_k\lVert + \bar{w} + d_u,
\qquad d_u:=L_u\bar{r}\Delta_u.
\end{align}
We can give the error bound of every step $l$ of the MPC recursively as $\lVert \bm{e}_{k+l}\lVert \leq (\bar{w}+d_u) S_{\hat{\tau}+l}$. As \eqref{eq:PredictionErrorCondition} requires \mbox{$\lVert \bm{e}_{k+l}\lVert+(1-\gamma)\lVert \bm{e}_{k+l-1}\lVert\leq \Delta_{k+l}^{\mathrm{rem}}$}, we have \mbox{$G(\hat{\tau},l) (\bar{w}+d_u) \leq \Delta^\mathrm{rem}_{k+l}(\hat{\bm{x}}_k,\hat{\tau},l)$},
which is equivalent to \eqref{eq:MPCCBFNoiseCondition}. Forward invariance follows from the same arguments as in the proof of Lemma \ref{lem:MyopiCBF}.
\end{proof}
We observe that the bound on the tolerated disturbance for the remote MPC-CBF depends on two additional factors: i) the error due to uncompensated delay $d_u$, reducing the disturbance margin directly, ii) the prediction error from compensating $\hat{\tau}$ and enforcing the condition along the MPC horizon.

\subsection{Architectural Disturbance Tolerance Comparison}
Comparing the disturbance bounds derived in Lemmas \ref{lem:MyopiCBF} and \ref{lem:MPCCBF} reveals that the remote placement yields a smaller admissible disturbance tolerance. We combine our findings in the following proposition, where we interpret the bounds with respect to an identical barrier margin at the time of safety enforcement.
\begin{proposition}[Architecture-Induced Disturbance Tolerance]
\label{prop}
Consider local and remote CBF constraints with identical available
CBF margin $\Delta>0$. Let
\begin{align*}
S_r:=\sum_{j=0}^{r-1}L_x^j,\qquad
G(\hat{\tau},l):=S_{\hat{\tau}+l}+(1-\gamma)S_{\hat{\tau}+l-1},
\end{align*}
with $S_0 = 0$ and $d_u:=L_u \bar{r}\Delta_u$. The admissible disturbance bounds
satisfy
\begin{align*}
\bar w_l (\cdot)=\Delta,
\qquad
\bar w_r(\hat{\tau},\bar{r},l)
=
\frac{\Delta}{G(\hat{\tau},l)}-d_u.
\end{align*}
Assume $\bar w_r>0$. Then:
\begin{enumerate}
\item For $\hat{\tau}=0$, $\bar{r}=0$, and $l=1$:
    \begin{align}
         \label{eq:propEquality}
    \bar w_r(0,0,1)=\bar w_l.
    \end{align}
\item For $\hat{\tau}=0$, $\bar{r}=0$, and $l>1$:
    \begin{align}
         \label{eq:propNoDelayButHorizon}
    \bar w_r(0,0,l)<\bar w_l.
    \end{align}
\item For $\hat{\tau}>0$ and $l\in \{1,\dots,N\}$:
    \begin{align}
     \label{eq:propDelayAndHorizon}
    \bar w_r(\hat{\tau},\bar{r},l)<\bar w_r(0,0,l)\leq \bar w_l.
    \end{align}
\end{enumerate}
\end{proposition}
\begin{proof}
For $\hat{\tau}=0$ and $l=1$, we have $S_1=1$ and $S_0=0$,
hence $G(0,1)=1$. If also $\bar{r}=0$, then $d_u=0$ and
$\bar w_r=\Delta=\bar w_l$.\\
For $l>1$, $S_l>1$ and $S_{l-1}\geq 1$, hence
$G(0,l)>1$, which implies $\bar w_r(0,0,l)<\Delta$.\\
If $\hat{\tau}>0$, then $G(\hat{\tau},l)>G(0,l)$ because additional positive
terms enter the sums $S_{\hat{\tau}+l}$ and $S_{\hat{\tau}+l-1}$. If
$\bar{r}>0$, then $d_u>0$ directly subtracts from the admissible
disturbance bound. Therefore, both compensated, known delay $\hat{\tau}$ and an uncompensated
delay residual $r_k$ reduce the remote disturbance tolerance.
\end{proof}
Proposition \ref{prop} formalizes the architecture-induced robustness loss of the remote placement. Remote MPC-CBF accumulates prediction errors through delay compensation and horizon propagation, so its admissible disturbance bound decreases with $\hat\tau$, $\bar{r}$, and $l$. Local CBFs, therefore, offer greater disturbance tolerance but remain myopic and may react late, leading to deadlocks or degraded performance. 
Remote MPC-CBF provides anticipatory behavior through the performance objective, but only under tighter admissible disturbances. This motivates the subsequent combined architecture.
\begin{table*}[t]
\caption{Average Results for 50 runs (L - Local CBF, R - Remote MPC-CBF, C - Combined Architecture)}
\label{tab:simulation_result}
\centering
\setlength{\tabcolsep}{4.8pt}
\footnotesize
\begin{tabular}{c|ccc|ccc|ccc|ccc|ccc|ccc}
        \toprule
        \multirow{2}{*}{$\bm{\tau =\hat{\tau} + \bar{r}}$}
        & \multicolumn{3}{c|}{\textbf{Reach Rate}}
        & \multicolumn{3}{c|}{\textbf{Safe Rate}}
        & \multicolumn{3}{c|}{\textbf{Median }$\bm{d_{\min}}$ [mm]}
        & \multicolumn{3}{c|}{\textbf{Avg. }$\bm{t_{\mathrm{reach}}}$ [s]}
        & \multicolumn{3}{c}{\textbf{Avg. Stage Cost}}
        & \multicolumn{3}{c}{$\overline{\delta\bm{u}}$ [Nm]}
        \\
        \cmidrule(lr){2-4}
        \cmidrule(lr){5-7}
        \cmidrule(lr){8-10}
        \cmidrule(lr){11-13}
        \cmidrule(lr){14-16}
        \cmidrule(lr){17-19}
        & L & R & C
        & L & R & C
        & L & R & C
        & L & R & C
        & L & R & C
        & L & R & C\\
        \midrule
        I: $\tau = 3 + 0$  &100\%  &100\%  &100\% &98\%  &92\%  &\textbf{100}\%  &0.02  &1.64  &\textbf{2.07}  &8.28  &7.98  &\textbf{7.94}  &14.51  &5.47  &\textbf{5.44} &3.73  &\textbf{0.22}  &0.23  \\
        II: $\tau = 2 + 1$ &98\%  &98\%  &98\%  &100\%  &90\%  &100\%  &0.03  &4.88  &\textbf{4.89}  &7.21  &7.12  &\textbf{7.10}  &20.41  &7.41  &\textbf{7.39}  &3.96  &\textbf{0.46}  &0.59  \\
        III: $\tau = 6 + 0$ &96\%  &96\%  &96\%  &96\%  &72\%  &\textbf{100}\%  &0.02  &\textbf{1.13}  &0.99  &\textbf{8.58}  &8.73  &8.71  &14.26  &\textbf{5.46}  &5.59  &3.95  &\textbf{0.29}  &0.62  \\
        IV: $\tau = 5 + 1$ &96\%  &98\%  &98\% &100\%  &64\%  &100\%  &0.02  &3.09  &\textbf{3.61}  &8.25  &\textbf{7.69}  &7.70  &18.90  &\textbf{6.96}  &8.31 & 4.08  &\textbf{0.49}  & 1.17  \\
        \bottomrule
    \end{tabular}
\end{table*}
\section{Combined Architecture}
\label{sec:CombinedArchitecture}
While both local and remote CBF placements ensure safety under suitable disturbance bounds, they exhibit complementary properties: local CBFs provide higher disturbance tolerance, whereas MPC-CBF improves performance through predictive planning. We therefore propose a combined architecture that leverages both mechanisms for a known disturbance bound $\bar{w}$ and delay $\tau$.
To this end, we robustify the local CBF by tightening the barrier condition as in \cite{Liu.2026}:

\begin{minipage}{0.98\linewidth}
\vspace{0.8ex}
    \hrule
    \vspace{0.8ex}
    \textbf{Robust Local CBF:}
    \vspace{1.2ex}
            \begin{align}
            \label{eq:RDCBF} 
            \bm{u}_k& = \argmin_{\bm{u}} \quad\frac{1}{2}\lVert \bm{u}-\bm{u}_k^\mathrm{nom} \rVert^2            \\
           \textrm{s.t. } h(\hat{\bm{x}}_{k+1})&-h(\bm{x}_k) \geq -\gamma h(\bm{x}_k)+L_h \bar{w},\, \bm{u} \in \mathcal{U}\nonumber
            \end{align}
    \vspace{-2.5ex}
    \hrule
    \vspace{1.5ex}
\end{minipage}
The resulting feasible set of the robust local CBF is
\begin{align}
    \label{eq:locCBFFeasibleSet}
    \mathcal{U}_\mathrm{CBF}(\bm{x}_k)
    \coloneqq
    \bigl\{\bm{u}\in\mathcal{U}\;:\;&
    h(\bm{f}(\bm{x}_k,\bm{u})) \notag\\
    &\geq (1-\gamma)h(\bm{x}_k)+L_h\bar{w}
    \bigr\}.
\end{align}
 As long as $\mathcal{U}_\mathrm{CBF} \neq \emptyset$ and $\lVert\bm{w}_k\lVert\leq \bar{w}$, the local CBF is feasible and thus guarantees forward invariance of the safe set. Notably, the feasibility of the local problem is not directly tied to that of the remote side, as the local CBF may be feasible even if the remote MPC is not. Similarly, the feasibility of the remote MPC for a predicted state $\hat{\bm{x}}_k$ does not guarantee
the feasibility of the local CBF for the true state $\bm{x}_k$ due to disturbances. Therefore, the remote MPC-CBF should not only be tightened with the propagating error margin \eqref{eq:MPCCBFNoiseCondition}, but its allowable input set should also be tightened as $\mathcal{U}_\mathrm{nom}\coloneqq\mathcal{U}\ominus\Delta U$. Thus, additional input margin is reserved for the local CBF. This procedure parallels tube-MPC.
Since the feasibility of the local CBF is sufficient for safety and the CBFs propagated along the horizon may significantly influence the feasibility of the remote MPC-CBF, we augment it with slack variables $\delta_i$.
\begin{subequations}
\label{eq:RMPC-CBF} \nonumber
\begin{minipage}{\linewidth}
\vspace{0.8ex}
\hrule
\vspace{0.8ex}
\textbf{Robust Remote MPC-CBF:}
\vspace{1.2ex}
        \begin{align}
         J = &\min_{\bm{u}_{0:N-1|k}} p(\bm{x_{N|k}})+\sum_{i=0}^{N-1}\left(q(\bm{x}_{i|k},\bm{u}_{i|k})+\rho \delta_i^2\right)\\
         &\textrm{Constraints \eqref{eq:OCP-dynConst}-\eqref{eq:OCP-inputConst} with additional constraints} \nonumber\\
        &h(\bm{x}_{i+1|k}) \geq 
            (1-\gamma) h(\bm{x}_{i|k})+  \nonumber \\
        & \qquad \qquad \quad+L_h G(\hat{\tau},i+1)(\bar{w}+d_u) -\delta_i,
        \label{eq:rmpc-cbf-barrier}\\ 
        &\delta_i \geq 0, \quad \rho \geq 0
        \label{eq:rmpc-cbf-slackconstraint}\\ 
        &\bm{u}_{i|k}\in\mathcal{U}_\mathrm{nom},\\
        & \forall i=\{0,\ldots,N-1\} \nonumber
        \end{align}
    \vspace{-2.5ex}
    \hrule
    \vspace{1.5ex}
    \end{minipage}
\end{subequations}
The required correction set $\Delta U$ depends on the difference in disturbance tolerance \eqref{eq:myopicNoiseCondition} and \eqref{eq:MPCCBFNoiseCondition} such that if $\bm{u}_k^\mathrm{nom} \in \mathcal{U} \ominus \Delta U$ then $\bm{u}_k-\bm{u}_k^\mathrm{nom} \in \Delta U$ and therefore $\bm{u}_k\in\mathcal{U}$. Its analytical derivation 
is left for future work.\\
The resulting architecture imposes a hierarchy: the remote MPC-CBF improves performance and provides a predictive safety certificate when its tightened constraints are feasible with \(\delta_i=0\). If the remote CBF constraints are relaxed, i.e. $\delta_i>0$, or the remote problem fails, safety is delegated to the local robust CBF. Thus, forward invariance is guaranteed under the local feasibility condition $\mathcal{U}_\mathrm{CBF} \neq \emptyset$. The tightened input set \(\mathcal U\ominus\Delta\mathcal U\) reserves control authority for this local correction. Uncompensated delay, e.g. in the downlink, affects the reference input supplied to the local filter, but not the timing of the safety certification itself.
Note that in the joint architecture, the predictor must account for the local CBF corrections to maintain consistency between predicted and applied inputs.

\section{Simulation Results}
\label{sec:sim}
\begin{figure}[b]
    \centering
    \includegraphics[width=.7\linewidth]{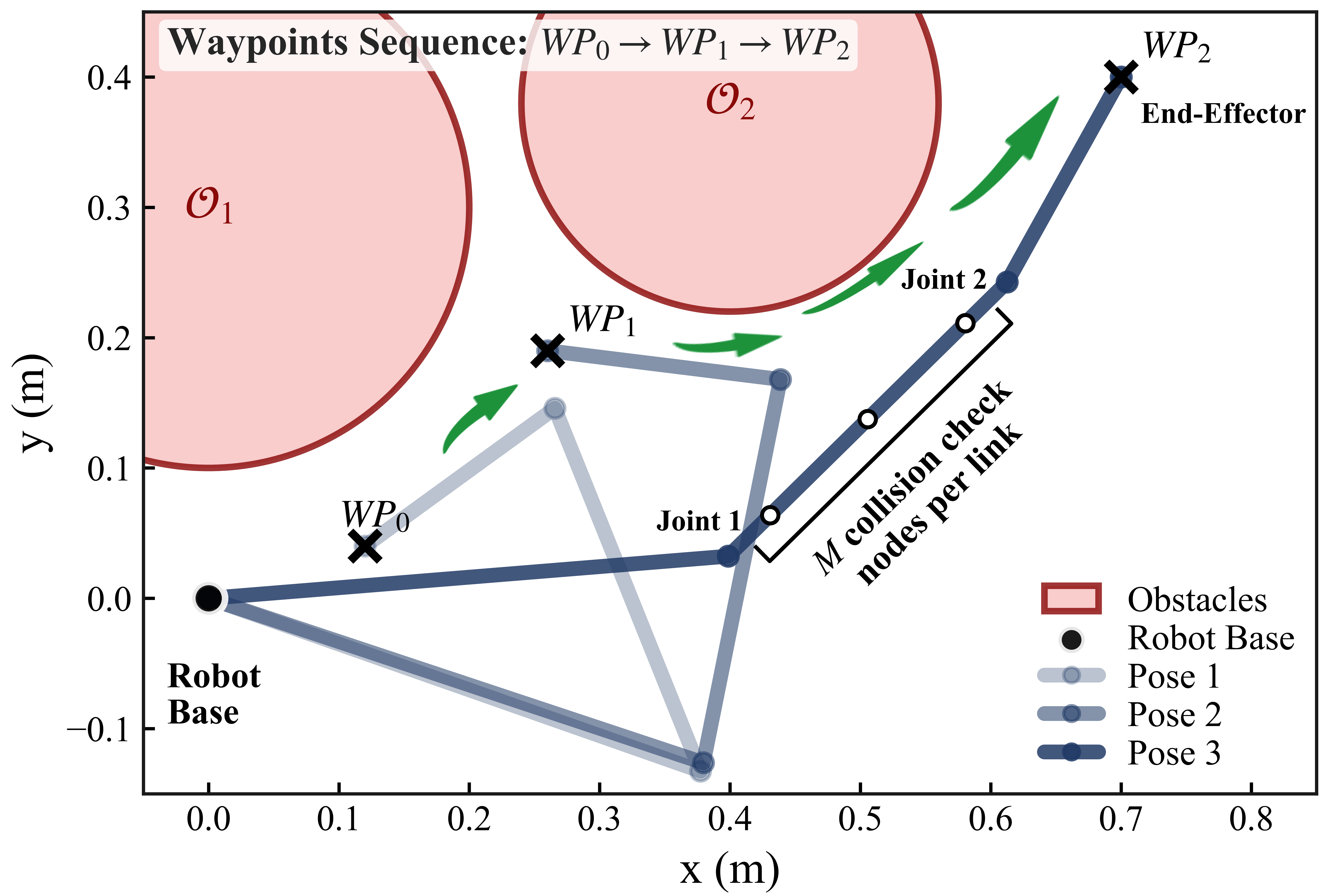}
    \caption{Considered planar 3-DoF collision avoidance task.}
    \label{fig:SimulationSetup}
\end{figure}

To illustrate the consequences of CBF placement in NCS, we consider a collision-avoidance task for a 3-DoF planar robot (cf. Fig.~\ref{fig:SimulationSetup}). We consider standard manipulator dynamics
\begin{align}
    \label{eq:SimDynamics}
    \bm{M}(\bm{q})\ddot{\bm{q}} + \bm{C}(\bm{q},\dot{\bm{q}})\dot{\bm{q}} +\bm{G}(\bm{q})= \bm{u} + \bm{\tau}_d
\end{align}
where $\bm{q} \in \mathbb{R}^3$ denotes the joint angles,  and $\bm{\tau}_d \in \mathbb{R}^3$ time-varying disturbances.
The robot's task is to track two waypoints in sequence while avoiding two circular obstacles. 
To ensure full-body safety, we discretize each link into three points. The corresponding barrier function of the $p$-th sample point of link $j$ with respect to obstacle $i$ is defined as a quadratic distance constraint
\begin{align}
    \label{eq:SimBarrier}
    h_{jp,i}(\bm{x}_k) = \lVert \bm{p}_{jp}(\bm{x}_k)- \bm{\mathcal{O}}_i\rVert^2 - (r_i+\epsilon)^2
\end{align}
with safety margin $\epsilon=0.005$m. The control is split into two components. A high-level MPC acts as a remote planner, computing a desired trajectory with $T_\mathrm{MPC}=20$Hz. Locally, a PD+-controller \cite{Murray.2017} tracks the desired path with a sampling rate of $200$Hz. The control barrier function acts on the planned trajectory, either as a constraint on the remote MPC or as a myopic filter locally before the PD+-controller. The simulations use the nominal CBF implementations to evaluate how the architectures behave under model mismatch and disturbances without conservative robust tightening. The plant is simulated with the full nonlinear rigid-body dynamics, whereas the MPC and CBF controllers use a simplified acceleration-level prediction model. Hence, the simulation includes model mismatch in addition to bounded time-varying disturbances. The latter are zero-mean Gaussian with $\sigma = 0.01$Nm, bounded by $\lVert\bm{\tau}_d\lVert \leq 0.02 \textrm{Nm}$.
For the MPC we consider a quadratic objective with respect to the way points as positional references, as well as a horizon $N=30$. Additionally, we constrain all joint accelerations to $\lvert\ddot{q}_j \rvert \leq20 \mathrm{rad/s^2}$.\footnote{The code is available at \href{https://github.com/TUM-ITR/WhereToPutSafety.git}{\texttt{github.com/TUM-ITR/WhereToPutSafety}}.}\\
First, we compare the three proposed architectures of this work, with respect to safety, median minimum safety margin $d_{\mathrm{min}}$, average completion time $t_\mathrm{reach}$, average stage costs as considered in the MPC, as well as average maximum torque command variation $\overline{\delta\bm{u}} =\lVert{\bm{u}_{k+1} -\bm{u}_k\lVert}_{\infty}$. Results over 50 runs for different delay cases I: $\hat{\tau}=3,\bar{r}=0$, II: $\hat{\tau}=2,\bar{r}=1$, III: $\hat{\tau}=6,\bar{r}=0$, \mbox{IV: $\hat{\tau}=5, \bar{r}=1$}  are summarized in Table \ref{tab:simulation_result}. All delays are reported in MPC sampling steps. The two pairs I/II and III/IV separate the effect of compensated delay $\hat{\tau}$ from the residual delay bound $\bar{r}$: increasing $\hat{\tau}$ primarily amplifies prediction error, while a nonzero $\bar{r}$ adds input-mismatch error.
We observe that the safe rate, i.e. runs with \mbox{$\min h_i(\bm{x}_k)> -10^{-5}\,\, \forall i,k$}, is fairly constant for the local CBF, while it drops significantly for the remote CBF with higher delay and thus more amplified errors. The change from 4 unsafe runs in scenario I to 18 unsafe runs in scenario IV clearly highlights the reduced disturbance tolerance of predictive safety under delay. However, the myopic nature of the local CBF can also be inferred, as the minimum distance to the safety boundary is two orders of magnitude smaller than for the predictive methods, even in the low delay cases, highlighting late intervention. The larger command variation $\overline{\delta\bm{u}}$ for the local CBF further indicates that local safety is enforced through later, more abrupt corrections.
 Regarding performance, the average stage cost is consistently much higher for the local CBF, whereas the remote and combined architectures maintain lower costs through predictive planning. 
 As intended by design, the combined architecture consistently achieves good performance while maintaining safety at all runs.
Fig. \ref{fig:StateSpacePlots} shows the average location and the hull of all trajectories from scenarios I and IV for the local CBF architecture and the remote MPC-CBF architecture.
Additionally, Fig. \ref{fig:statisticsOverDelay} depicts the median and the min-max envelope of the minimum distance to obstacle $2$ over $10$ runs for the three architectures over different constant delays $\hat{\tau}$. The remote MPC-CBF has a reasonable median clearance, yet its lower envelope becomes negative for larger delays, indicating unsafe runs caused by amplified prediction errors. In contrast, the local CBF remains close to the boundary with little variation, indicating myopic but safe behavior. The combined architecture maintains a positive clearance over all tested delays while providing a larger median distance than the local CBF.\\
\begin{figure}[t]
\centering
\includegraphics[width=\columnwidth]{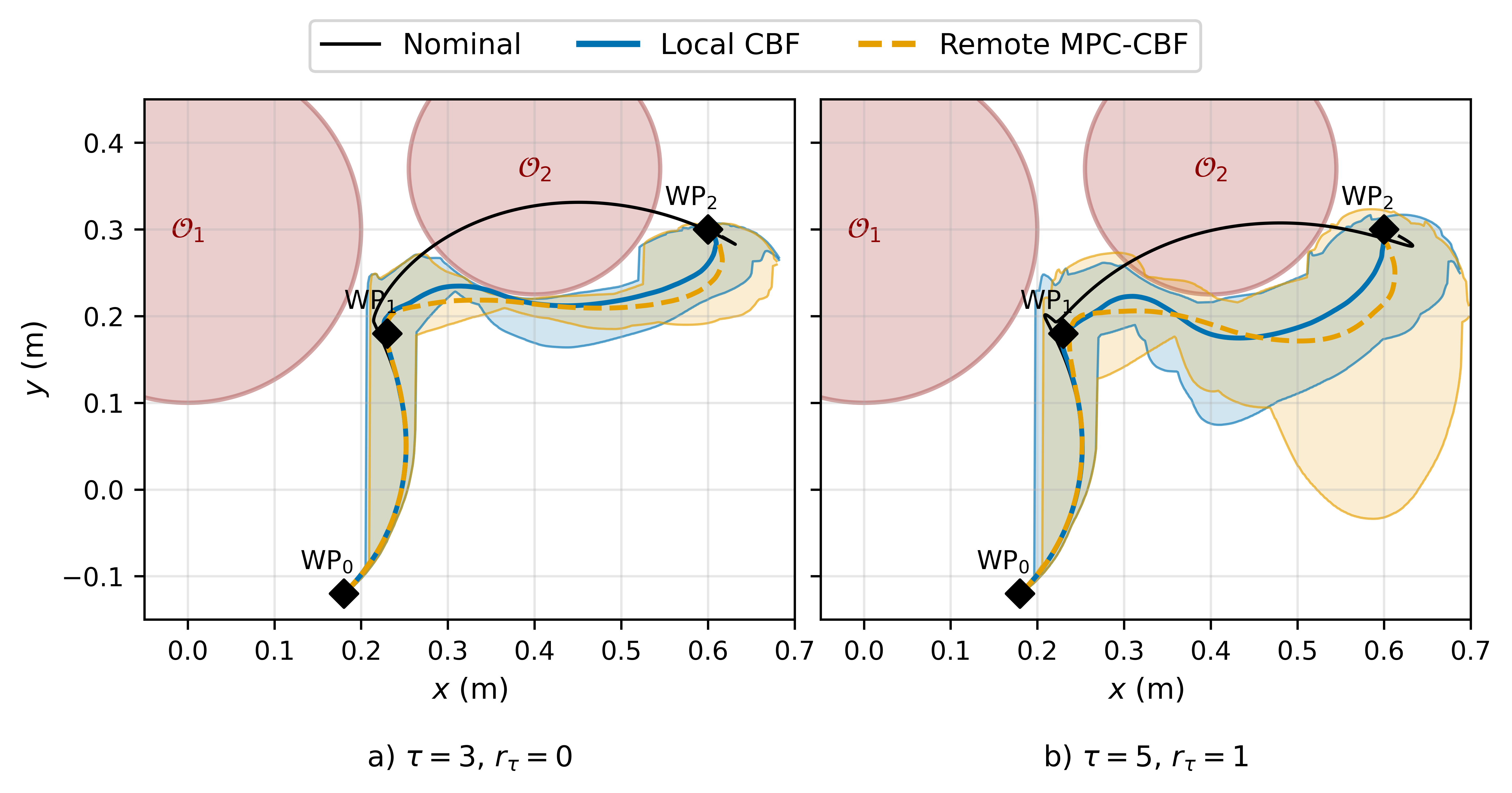}
\caption{End-effector paths over 50 runs for $\tau = 3 + 0$ and $\tau = 5 + 1$. Solid lines denote spatially averaged trajectories and shaded regions the trajectory hull.}
\label{fig:StateSpacePlots}
\end{figure}
Overall, the results confirm the qualitative trade-off between robustness and performance induced by CBF placement.

\section{Conclusion}
\label{sec:conclusion}
This work shows that the placement of CBFs in NCS fundamentally affects both safety and performance. We derived disturbance-tolerance bounds for local and remote CBF implementations and showed that predictive safety under delay yields stricter admissible disturbance levels.
Simulations confirm that local CBFs provide superior robustness due to access to fresh state information, while remote MPC-CBF improves performance through anticipatory behavior, yet is more sensitive to disturbances. The proposed combined architecture provides a practical compromise in uncertain networked environments.
Future work includes analyzing the implications of safety filter placement on the feasibility of the optimization problems, more complex disturbances, and realistic network settings, as well as validating the proposed architecture experimentally.

\section*{DECLARATION OF GENERATIVE AI AND AI-ASSISTED TECHNOLOGIES IN THE WRITING PROCESS}
During the preparation of this work, the authors utilized ChatGPT and Grammarly to enhance the writing, simulation, and technical content. After using these tools, the authors reviewed and edited the publication's content as needed and take full responsibility for it.
\appendices
\begin{figure}[t]
    \centering
    \includegraphics[width=.8\linewidth]{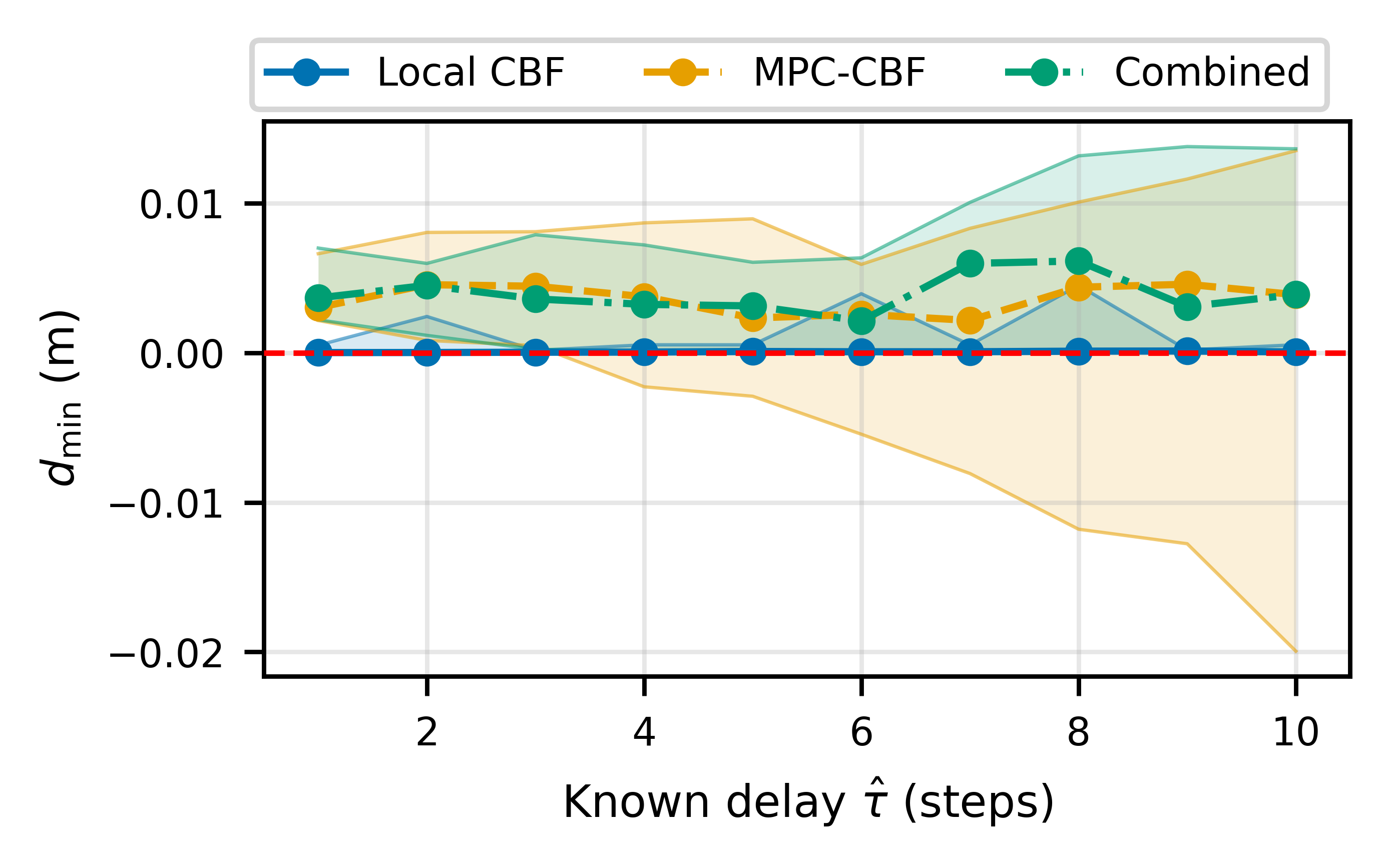}
    \caption{Minimum distance to the safety boundary of obstacle $2$ over a range of known delays $\hat{\tau}$ with $r_k=0$. Lines show the median value, shaded regions show the min-max envelope over 10 runs.}
    \label{fig:statisticsOverDelay}
\end{figure}
\printbibliography

\end{document}